\documentclass[12pt]{article}
\usepackage[centertags]{amsmath}
\usepackage{amsfonts,amsthm,amssymb}
\usepackage{amssymb}
\usepackage{amsmath}
\usepackage{graphicx}
\usepackage{MnSymbol}
\usepackage{color}
\usepackage[margin=1.00in]{geometry}
\theoremstyle{definition}
\newtheorem{theorem}{Theorem}
\newtheorem{corollary}{Corollary}
\newtheorem{definition}{Definition}
\newtheorem{lemma}{Lemma}
\newtheorem{proposition}{Proposition}
\newtheorem{remark}{Remark}

\newtheorem{example}{Example}

\DeclareMathOperator{\conv}{conv}
\DeclareMathOperator{\avg}{avg}

\begin{document}

\title{Pareto Efficient Nash Implementation Via Approval Voting}
\author{Yakov Babichenko\thanks{Department of Industrial Engineering and Management, Technion, Israel Haifa 3200003. This research was supported by the Air Force Office of Scientific Research grant \# FA9550-09-1-0538, by the Office of Naval Research grant \# N00014-09-1-0751, and by a Walter S. Baer and Jeri Weiss fellowship.}\ ,\ Leonard J.\ Schulman\thanks{Department of Computing and Mathematical Sciences, Caltech. Pasadena CA 91125. This research was supported in part by NSF Awards \# 1319745, 1618795.}}

\maketitle

\begin{abstract}

We study
implementation of a social choice correspondence in the case of two players who have von Neumann - Morgenstern utilities over a finite set of social alternatives, and the mechanism is allowed to output lotteries. Our main positive result shows that a close variant of the popular \emph{approval voting} mechanism succeeds in selecting \emph{only} Pareto efficient alternatives as pure Nash equilibria outcomes. Moreover, we provide an exact characterization of pure Nash equilibria profiles and outcomes of the mechanism. The characterization demonstrates a close connection between the approval voting mechanism and the notion of \emph{average fixed point,} which is a point that is equal to the average of all points that it does not Pareto dominate.


\end{abstract}

\section{Introduction}

Pareto efficiency is a key requirement in mechanism design and implementation theory. 
The question of whether mechanisms \textit{exist} that implement Pareto efficient outcomes, has been studied extensively in the literature. The classical works of Maskin~\cite{Ma} and others~\cite{DS,MR,AS} characterise Nash implementable social choice correspondences (SCCs) and construct robust mechanisms, so-called \emph{canonical mechanisms}, that accomplish implementation of \emph{every} Nash implementable SCC (in particular, Pareto efficient ones). These canonical mechanisms, however, are complex, and not likely to be used in a practical setting. 
This paper focuses on the question whether Pareto efficiency can be implemented by (non-dictatorial) \emph{simple and realistic mechanisms}.
We address this question in a quite general setting: two players who have von Neumann - Morgenstern utilities over a finite set of social alternatives\footnote{The results can be generalized to the case of ordinal preferences, see Section \ref{sec:ord}.}. Building upon the approval voting mechanism that is simple and widely used in practice (see~\cite{BF,ZMP}), we introduce a mechanism that implements Pareto efficient outcomes only.

We consider a standard setting that the number of alternatives is finite, and the mechanism's outcomes are \emph{lotteries} over the alternatives\footnote{This exact setting is addressed in virtual implementation, see~\cite{AS}.}. We are looking for a mechanism that satisfies the following properties:
\begin{itemize}
\item[(a)] A pure Nash equilibrium\footnote{Alternatively we may focus on the subset of strict Nash equilibria. All the results in the paper hold for strict equilibria, under mild genericity assumptions, see Section \ref{sec:str}.} always exists.
\item[(b)] All pure Nash equilibria outcomes are (approximately) Pareto efficient.
\item[(c)] The mechanism is anonymous: namely, it is symmetric in the players (this requirement reflects ``fairness" of the mechanism).
\end{itemize}

In the \emph{approval voting} mechanism each player either \emph{approves} or \emph{disapproves} each of the alternatives. The alternative that achieves the maximum number of approvals is chosen. In case of a tie the choice is uniform at random among the maximizers\footnote{The case of non-uniform choice among the maximizers is addressed in Section \ref{sec:non-unif}.}. The approval voting mechanism has been used for decades by many professional societies\footnote{For instance, the Institute of Electrical and Electronics Engineers (IEEE),  the Mathematical Association of America (MAA), and the American Statistical Association (ASA), to mention a few.}~\cite{BF} and more recently by the commercial coordination program Doodle, see~\cite{ZMP}. It is known that the approval voting mechanism admits Pareto efficient pure Nash equilibrium, but it also might admit inefficient equilibria, see \cite{BF} and Example \ref{ex:inef}. We consider the following modification of the approval voting mechanism: If any alternatives are approved by both players, then with probability $(1-\delta)$ the mechanism chooses uniformly at random among the alternatives that achieve two approvals (as in the original mechanism); while with probability $\delta$ the mechanism chooses uniformly at random among the alternatives that are approved by at least one player. A motivation for such a modification of the mechanism appears in Remark \ref{rem:mot}. Note that the mechanism requires only a binary message regarding each alternative---this is simpler than the direct revelation principle, which requires the exact utility at each alternative. Our main positive result (see Section \ref{sec:eq-char}) provides an exact characterization of pure Nash equilibria outcomes  and profiles of the modified approval voting mechanism. This characterization implies that the modified approval voting mechanism satisfies the above properties (1)-(3). Interestingly, the outcomes of the mechanism are closely related to the notion of an \emph{average fixed point} (see Section \ref{sec:afp}) which is a point that is equal to the average of all points that it does not Pareto dominate. Roughly speaking, the set of outcomes is the set of Pareto efficient alternatives that Pareto dominate an average fixed point, or an average fixed point itself if it is Pareto efficient.

The paper is organized as follows. Section \ref{sec:rl} further discusses the related literature on approval voting and implementation theory. The main results are in Section \ref{sec:av}, where we also discuss the application of the results to bargaining (Section \ref{sec:bar}). Section \ref{sec:non-unif} generalizes the results for a wider class of mechanisms, and Section \ref{sec:dis} concludes with a discussion.
  
\subsection{Related Literature}\label{sec:rl}

\subsubsection{Approval Voting}

Approval voting has been studied in the voting context~\cite{BF,B,MW}, generally in the case where the number of voters is large and the number of alternatives is small. It has been observed that approval voting admits a Pareto efficient Nash equilibrium~\cite{BF}, but also that it might admit Pareto inefficient equilibria~\cite{BF}. To the best of our knowledge,
our work is the first in which
a reasonable modification of the approval voting succeeds in selecting \emph{only} Pareto efficient outcomes as pure Nash equilibria outcomes. We do this for two-player interactions; in Section~\ref{sec:3p} we show it to be impossible for three or more players.

Most closely related to ours is a 
recent paper by N\'u\~nez and Laslier~\cite{NL}, which demonstrates efficiency of the approval voting mechanism in two-player interactions by showing the following result: Under the behavioural assumption of \emph{partial honesty}\footnote{The partial honesty assumption is that if a player is indifferent between two actions $a$ and $b$ where $a$ is a \emph{sincere} action (namely it contains all alternatives above some threshold) and $b$ is not sincere, then the player strictly prefers $a$ over $b$.} the two-player game that is induced by the approval voting mechanism admits a pure Nash equilibrium and all its pure Nash equilibria are Pareto efficient. Our results differ from \cite{NL} because we do not impose any behavioural assumption on the players, namely players preferences are determined solely by their utilities.
Moreover, our lower bounds on the utilities of the players in the modified approval voting mechanism improve upon the lower bounds of the standard approval voting mechanism. In \cite{NL} it is proven that every player gets at least the average outcome. We prove that every player gets at least his utility in the worst average fixed point. Every average fixed point Pareto dominates the average outcome, so this is a stronger bound. 
Finally, in the context of partial honesty and sincereness, the main theorem in \cite{NL} shows that for the standard approval voting mechanism there exists a pure Nash equilibrium where both players use sincere actions. It is interesting to note that in our modified approval voting mechanism this statement
is no longer true, see Example \ref{ex:sin} in Section \ref{sec:seq-elim}.


\subsubsection{Implementation Theory}

Implementation of Pareto efficient social choice correspondences (SCCs) is a central topic in implementation theory. The impossibility result by Hurwicz and Schmeidler~\cite{HSh} states that only dictatorial SCCs can be Nash implemented if we require that a pure Nash equilibrium always exists, and that all pure Nash equilibria outcomes are Pareto efficient. This impossibility result can be avoided by either posing (mild) domain restrictions on the preferences of the players (see e.g.,~\cite{DS}) or by \emph{virtual implementation}, which requires only an approximate implementation of the social choice correspondence (see e.g.,~\cite{AS}). In this paper we stick to the virtual implementation approach: we have no restrictions on the preferences, but the mechanism implements outcomes that are $\varepsilon$-close to Pareto efficiency rather than exactly Pareto efficient.

More generally, Nash Implementation in complete information settings has been studied extensively in various aspects. One aspect, originated by Maskin~\cite{Ma}, characterizes the class of SCCs that can be Nash implemented, see~\cite{Ma,DS,MR,AS}. The canonical mechanisms that implement any implementable SCC have been criticized for their complexity which makes their usage somewhat unrealistic. Simplifications of the canonical mechanism have been studied in~\cite{Sa,Mc,Se}, but yet the canonical mechanism remains quite involved. This led to another branch of implementation theory which aims to implement concrete SCCs using simple and practically realizable mechanisms. Such mechanisms (beyond the dictatorial ones) that achieve Pareto efficiency in all equilibria, are known only in concrete cases such as exchange economies settings~\cite{Hu,DSV,Sj,STY}, and more specifically for Walrasian allocations~\cite{RR,Ch} and Lindahl allocations~\cite{Ti,Ti93}.
This paper belongs to this branch of implementation theory, and addresses the setting of a finite number of alternatives with lottery outcomes. This exact setting has been studied in the context of general SCCs in the works~\cite{AS,DS}, which however use canonical mechanisms; they show that such mechanisms may implement Pareto efficient outcomes only. The contribution of the present paper is in the simplification of the implementation mechanism. We show that a modified approval voting mechanism virtually implements a Pareto efficient SCC. Moreover, in Section \ref{sec:non-unif} we introduce a class of \emph{weighted approval voting mechanisms} which implements a richer class of Pareto efficient SCCs (see Corollary \ref{cor:scc}). Our class of SCCs is similar to the class of SCCs that was introduced by Dutta and Sen in~\cite{DS}. The class of SCCs in~\cite{DS} was proven to be implementable by the (relatively involved) canonical mechanism. Our class of SCCs is proven to be implementable by the (relatively simple) weighted approval voting mechanism.

\section{Approval Voting Mechanism}\label{sec:av}

The finite set of alternatives is $[n]=\{1,2,...,n\}$. Player's $i=1,2$ von Neumann - Morgenstern utility of the alternative $k=1,2,...,n$ is denoted by $a_i^k$. We normalize the utilities such that $0\leq a^k_i \leq 1$. The utility profile of alternative $k$ is denoted by $a^k:=(a^k_1,a^k_2)$. The collection of utility profiles is denoted by $A:=\{a^k:1\leq k \leq n\}$ which is a collection of $n$ points in $[0,1]^2$.

In the \emph{approval voting} mechanism ($\mathcal{AV}$) each player $i=1,2$ submits a set of approved alternatives $L_i\subset [n]$. If the sets are disjoint (i.e., $L_1 \cap L_2 =\emptyset$) then we say that players \emph{disagree}, and the outcome lottery is a uniform distribution over $L_1 \cup L_2$ (or over $[n]$ if $L_1 \cup L_2=\emptyset$). Otherwise, when $L_1 \cap L_2 \neq \emptyset$,  we say that players \emph{reach an agreement}, and the outcome lottery is the uniform distribution over $L_1 \cap L_2$.  Formally, we denote by $UN(B)$ the uniform distribution over a finite set $B$, and we define outcome lottery by
 
\begin{align*}
f(L_1,L_2)=\begin{cases}
UN([n])         & \text{ if } L_1=L_2=\emptyset, \\
UN(L_1\cup L_2) & \text{ if } L_1 \cap L_2 = \emptyset \text{ and } L_1\cup L_2\neq \emptyset, \\
UN(L_1 \cap L_2) & \text{ if } L_1 \cap L_2 \neq \emptyset,
\end{cases}
\end{align*}
where the first condition is needed only in order that $f$ will be well defined for all pairs of sets.

For two vectors $x,y\in [0,1]^2$ we say that \emph{$x$ (strictly) Pareto dominates $y$}, if $x_1>y_1$ and $x_2>y_2$, and we denote $x>>y$. We say that an outcome $x\in [0,1]^2$ is \emph{$\varepsilon$-Pareto efficient (with respect to $A$)} if there is no $a\in A$ such that $a>>x+(\varepsilon,\varepsilon)$. For $\varepsilon=0$ we say that $x$ is Pareto efficient.

The approval voting mechanism might contain inefficient equilibria as demonstrated in the following example.

\begin{example}\label{ex:inef}
In the case of $4$ alternatives, let the collection of utility profiles be 
\begin{align*}
A=((1,0),(0,1),(0.9,0.9),(\frac{2}{3},\frac{2}{3})).
\end{align*}
It is easy to check that the pure action profile $(L_1,L_2)=(\{1,4\},\{2,4\})$ is a pure Nash equilibrium with the outcome $(\frac{2}{3},\frac{2}{3})$ which is inefficient.
\end{example}

By considering the above example in more detail, we can see that both players do not lose by approving also the efficient alternative $(0.9,0.9)$. However, since the opponent disapprove this alternative neither player gains from approving it either.

In order to resolve this problematic issue, we provide to each player an incentive to approve this efficient alternative \emph{irrespective} of whether the opponent approves it or not. We consider the mechanism $\mathcal{AV}_\delta$ which is identical to the $\mathcal{AV}$ mechanism except in one aspect. In the case of agreement ($L_1 \cap L_2 \neq \emptyset$) the outcome lottery is the uniform distribution over $L_1 \cap L_2$ with probability $1-\delta$ (not with probability 1), and the uniform distribution over $L_1 \cup L_2$ with probability $\delta$. Formally,
\begin{align*}
f(L_1,L_2)=\begin{cases}
UN([n])         & \text{ if } L_1=L_2=\emptyset, \\
UN(L_1\cup L_2) & \text{ if } L_1 \cap L_2 = \emptyset \text{ and } L_1\cup L_2\neq \emptyset, \\
(1-\delta) UN(L_1 \cap L_2) +  \delta UN(L_1\cup L_2) & \text{ if } L_1 \cap L_2 \neq \emptyset,
\end{cases}
\end{align*}
\begin{remark}\label{rem:mot}
Such a ``mistake" with probability $\delta$ may have several explanations.
\begin{itemize}
\item (An obvious explanation.) An intentional ``mistake" by the mechanism designer in order to induce better behaviour.
\item (An alternative explanation.) The mechanism performs the standard approval voting mechanism, \emph{but}, there in a noise in the communication between the players and the mechanism. When player $i$ sends a $0/1$ message regarding whether he disapprove/approve the $j$'th alternative, with small probability $\varepsilon$ the message is flipped during the communication. If we assume that the noise is independent across alternatives and across players, then the resulting lottery is similar to the lottery of $\mathcal{AV}_\delta$ (when $\delta=\Theta(\varepsilon)$). The only difference is that in this latter settings with small probability (of order $O(\varepsilon^2)$) the outcome lottery will be the uniform distribution over $[n]$. Such a difference obviously does not effect the strategic reasoning. Therefore, all the results that will be developed for the mechanism $\mathcal{AV}_\delta$ will hold also in the scenario where the implemented mechanism is the approval voting but with communication noise. Interestingly, we will see that insertion of such a noise to the system results in a selection of Pareto efficient equilibria only\footnote{
The phenomenon that an insertion of a noise to a system may serve as an effective toll for selection of desirable equilibria is well known. Classical results establish such selections of equilibria in various settings~\cite{Sel,HS,Y,Sam}.
}.
\end{itemize}
\end{remark}

Before we state our main positive result, which is an exact characterization of the pure Nash equilibria outcomes of the mechanism $\mathcal{AV}_\delta$, we introduce a fixed-point notion which (as we will see in the results) is closely related to the mechanism $\mathcal{AV}_\delta$.

\subsection{Average fixed Point}\label{sec:afp}

We start with several notations. For a finite collection of vectors $B\subset[0,1]^2$ we denote $\avg(B)=\frac{1}{|B|}\sum_{b\in B} b$, where $\avg(B)\in [0,1]^2$.

For the collection of utility profiles $A$ and a point $x\in [0,1]^2$ we denote by $A_{<,<}(x)=\{a\in A: a_1<x_1, a_2<x_2\}$ the set of utility profiles in the third quadrant of the axis starting at $x$. Similar notation is used for the other quadrants (e.g., $A_{<,>}(x)$ is the second quadrant). Similarly, we denote $A_{\leq,\leq}(x)=\{a\in A: a_1\leq x_1, a_2 \leq x_2\}$ the third quadrant that includes the axis. 

\begin{definition}
A vector $x=(x_1,x_2)$ is a \emph{boundaries-included average fixed point} of the collection $A$ if $\avg(A\setminus A_{<,<}(x))=x$. 
\end{definition}

We call the average fixed point boundaries-included because the utility profiles on the lower boundary $A_{\leq,\leq}(x) \setminus A_{<,<}(x)$ are included in the computation of the average.

A slightly less restrictive notion of average fixed point, allows a situation where part of the points on the lower boundary belong to the averaging set and part do not.

\begin{figure}[h]
\begin{center}
\includegraphics[scale=1]{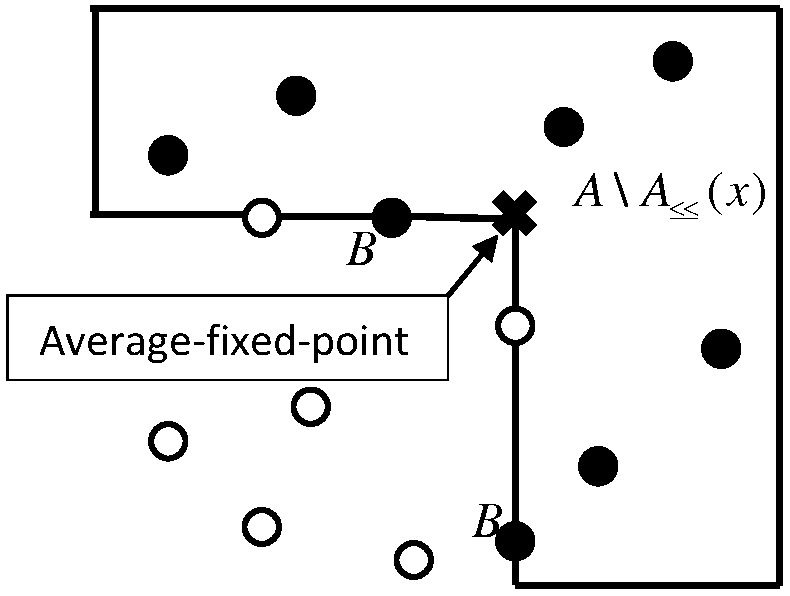}
\end{center}
\caption{An average fixed point.}\label{pic:afp}
\end{figure}

\begin{definition}
A vector $x=(x_1,x_2)$ is an \emph{average fixed point} of the collection $A$ if there exists a subset $B\subset A_{\leq, \leq}(x)\setminus A_{<,<}(x)$ such that $\avg((A\setminus A_{\leq, \leq}(x)) \cup B )=x$. We denote by $AFP(A)$ the set of all average fixed points.
\end{definition}

Figure~\ref{pic:afp} demonstrates the definition of an average fixed point. 

The following lemma demonstrates that even the more restrictive notion of boundaries-included average fixed point always exists, which obviously guarantees the existence of an average fixed point (i.e., $AFP(A)\neq \emptyset$).

\begin{lemma}\label{lem:afp}
Every collection $A$ admits at least one boundaries-included average fixed point.
\end{lemma}

\begin{proof} 
We set $x^1=\avg(A)$, and for $t\geq 2$ we set $x^t=\avg(A\setminus A_{<,<}(x^{t-1}))$. 

If  $A_{<,<}(x^1)=\emptyset$ then $x^1$ is a boundaries-included average fixed point. Otherwise, we know that $x^2>>x^1$ because only strictly-below-average utility profiles were eliminated from the set $A$. Similarly, if 
$A_{<,<}(x^2)\setminus A_{<,<}(x^1)=\emptyset$
 then $x^2$ is a boundaries-included average fixed point. Otherwise, $x^3>>x^2$ because only strictly-below-average utility profiles were eliminated from the set $A\setminus A_{<,<}(x^1)$. There are at most $n$ different outcomes, therefore for some $t\leq n+1$ we will have $A_{<,<}(x^{t})\setminus A_{<,<}(x^{t-1})=\emptyset$ and $x^t$ is a boundaries-included average fixed point.
\end{proof}

A natural question arises: Is a boundaries-included average fixed point necessarily unique? The following example demonstrates that the answer is no.

\begin{example}
Let $A=((1,1),(0.8,0),(0,0.8))$, then both $(1,1)$ and $(0.6,0.6)$ are boundaries-included average fixed points.
\end{example}

%

The set of average fixed points is not necessarily a singleton. However, the following lemma shows that the set of average fixed points has the following structure: it must be a sequence of (weakly) Pareto dominating outcomes.

\begin{lemma}\label{lem:afp-pd}
For every two average fixed points $x,y\in AFP(A)$, either $x\leq \leq y$ or $y\leq \leq x$.
\end{lemma}

\begin{proof}
Assume by way of contradiction that $x,y$ satisfies $x_1>y_1$ and $x_2<y_2$. Let 
$B_x\subseteq A_{\leq, \leq} (x) \setminus A_{<,<}(x) $ 
be such that $x=\avg((A\setminus A_{\leq, \leq} (x)) \cup B_x)$. We denote $A_x = (A\setminus A_{\leq, \leq} (x)) \cup B_x$ and we denote $B_x^C = (A_{\leq, \leq}(x) \setminus A_{<,<}(x))\setminus B_x$ the complementary lower boundary points (which are not included in the averaging). Similarly we denote $A_y$, $B_y$ and $B_y^C$.

\begin{figure}[h]
\begin{center}
\includegraphics[scale=1.2]{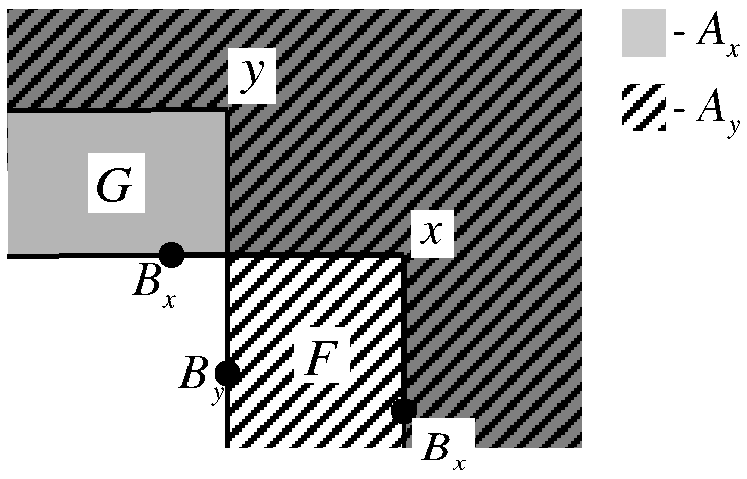}
\end{center}
\caption{The sets $A_x,A_y,B_x,B_y,F$ and $G$.}\label{pic:fg}
\end{figure}

Note that
\begin{align}\label{eq:ay}
A_y = (A_x \cup F) \setminus G
\end{align}
where (see Figure~\ref{pic:fg})
\begin{align*}
F &= \{a\in A: y_1\leq a_1\leq x_1, a_2\leq x_2\}\setminus B_y^C \text{ and}\\
G &= \{a\in A: a_1\leq y_1, x_2\leq a_2 \leq y_2\}\setminus B_y. 
\end{align*}

With respect to the average $x$ (and specifically the average $x_2$ of player 2), when we switch from $A_x$ to $A_y$, we add points that are weakly below $x_2$ (the set $F$), and we remove points that are weakly above $x_2$ (the set $G$). Therefore, $y_2\leq x_2$, which is a contradiction.
\end{proof}

The existence of an average fixed point (Lemma \ref{lem:afp}) and the structure of average fixed points (Lemma \ref{lem:afp-pd}) allows us to define the notion of \emph{minimal average fixed point}.

\begin{definition}\label{def:mafp}
A point $m_{afp}(A) \in [0,1]^2$ is a minimal average fixed point if $m_{afp} \in AFP(A)$, and for every $y\in AFP(A)$ it holds that $y \geq \geq m_{afp}$.
\end{definition}

Note that Lemmas \ref{lem:afp} and \ref{lem:afp-pd} prove existence and uniqueness of $m_{afp}(A)$.

%

\subsection{Characterization of the pure Nash equilibria}\label{sec:eq-char}

Before the statement of our main positive result, we introduce several notions.

For a collection $A$, we denote by $PE(A)=\{a\in A: \text{ there is no } b\in A \text{ such that } b >> a \}$ the set of Pareto efficient points of $A$. 

Our main positive result is an exact characterization of pure Nash equilibria outcomes of the mechanism $\mathcal{AV}_\delta$.

\begin{theorem}\label{theo:laa}
For every collection $A$ and for every $0< \delta \leq 1$, the set of pure Nash equilibria outcomes of the mechanism  $\mathcal{AV}_\delta$ is exactly the union of the following two types of outcomes:
\begin{enumerate}
\item The set of agreement outcomes 
\begin{align*}
AG(A)=\{(1-\delta)a+\delta x: a\in PE(A), x \in AFP(A), \text{ and } a\geq \geq x\}.
\end{align*}
\item The set of disagreement outcomes 
\begin{align*}
DIS(A)=\{x\in AFP(A): \text{There is no } a\in A \text{ such that } a >> x\}.
\end{align*}
\end{enumerate}
\end{theorem}
The equilibria outcomes are demonstrated in Figure ~\ref{pic:mt}. A collection $A$ with a single disagreement equilibrium, and two agreement equilibria is presented.

\begin{figure}[h]
\begin{center}
\includegraphics[scale=1]{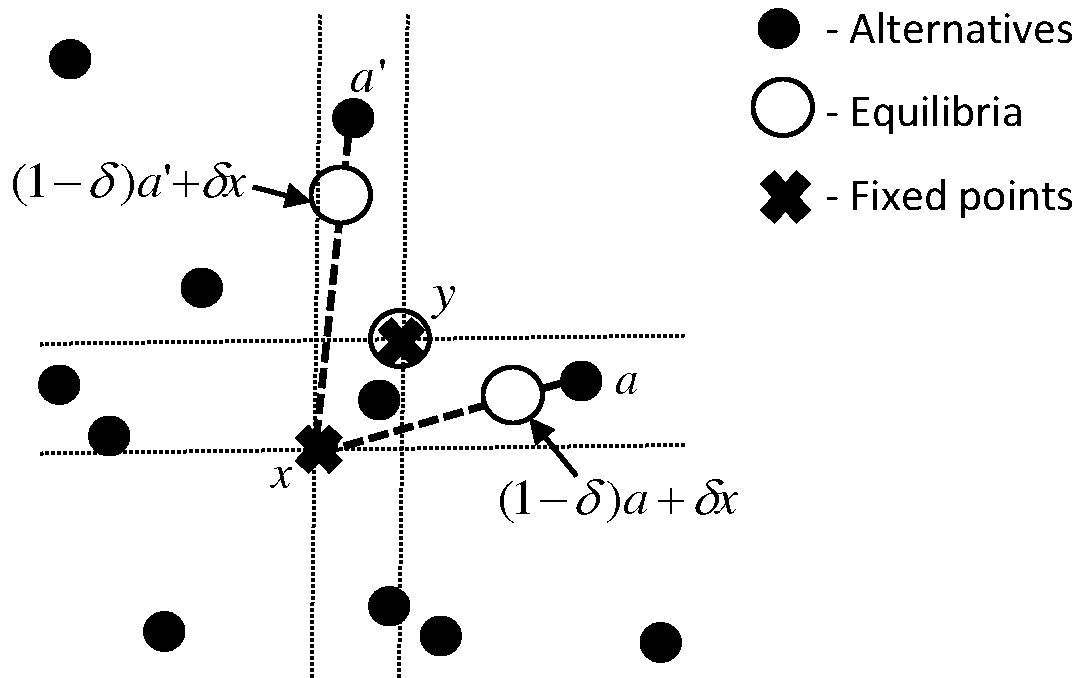}
\end{center}
\caption{Equilibria outcomes of $\mathcal{AV}_\delta$.}\label{pic:mt}
\end{figure}

The proof is presented in Section \ref{sec:proof}. A straightforward corollary shows that $\mathcal{AV}_\delta$ is indeed an anonymous mechanism \emph{all} of whose equilibrium outcomes are (approximately) efficient.

\begin{corollary}\label{cor:laa}
The mechanism $\mathcal{AV}_\delta$ admits a pure Nash equilibrium for every collection $A$, and $\mathcal{AV}_\delta$ is an anonymous mechanism which is $\delta$-Pareto efficient in all equilibria.
\end{corollary}

\begin{proof}[Proof of Corollary \ref{cor:laa}]
It is easy to check that $\mathcal{AV}_\delta$ is anonymous. It is also easy to check that all the elements in $AG(A)\cup DIS(A)$ are $\delta$-Pareto efficient. The remaining part is to show that $AG(A)\cup DIS(A)\neq \emptyset$. By Lemma \ref{lem:afp} $AFP(A)\neq \emptyset$. For some average fixed point $x\in AFP(A)$, if $x$ is not Pareto dominated by any utility profile, then $x\in DIS(A)$. Otherwise, there exists an $a\in PE(A)$ that Pareto dominates $x$, and then $(1-\delta)a+\delta x \in AG(A)$.
\end{proof}

Theorem \ref{theo:laa} can be stated in terms of social choice correspondences and virtual implementation. For a collection of utility profiles $A$ we define the \emph{minimal-AFP Pareto efficient correspondence} (see Definition \ref{def:mafp}) to be 
\begin{align*}
\sigma(A)=\{a\in PE(A \cup AFP(A)): a\geq \geq m_{afp}(A)\}.
\end{align*}
An immediate Corollary from Theorem \ref{theo:laa} is the following.

\begin{corollary}\label{cor:scc}
The mechanism $\mathcal{AV}_\delta$ virtually implements the social choice correspondence $\sigma$.
\end{corollary}



Actually, in the proof of Theorem \ref{theo:laa}, we show a stronger result that characterizes the set of Nash equilibria \emph{action-profiles} (not only outcomes).

For a collection $A$, the set of Nash equilibria action profiles of the game $\Gamma_{\mathcal{AV}_\delta}(A)$ is the union of the following two types of equilibria:

\textbf{Agreement equilibria}, which exist if there exists a Pareto dominated average fixed point $x\leq \leq a$ for $a\in PE(A)$. The equilibrium action profile is demonstrated in Figure \ref{pic:ag-profile}.

\begin{figure}[h]
\begin{center}
\includegraphics[scale=0.9]{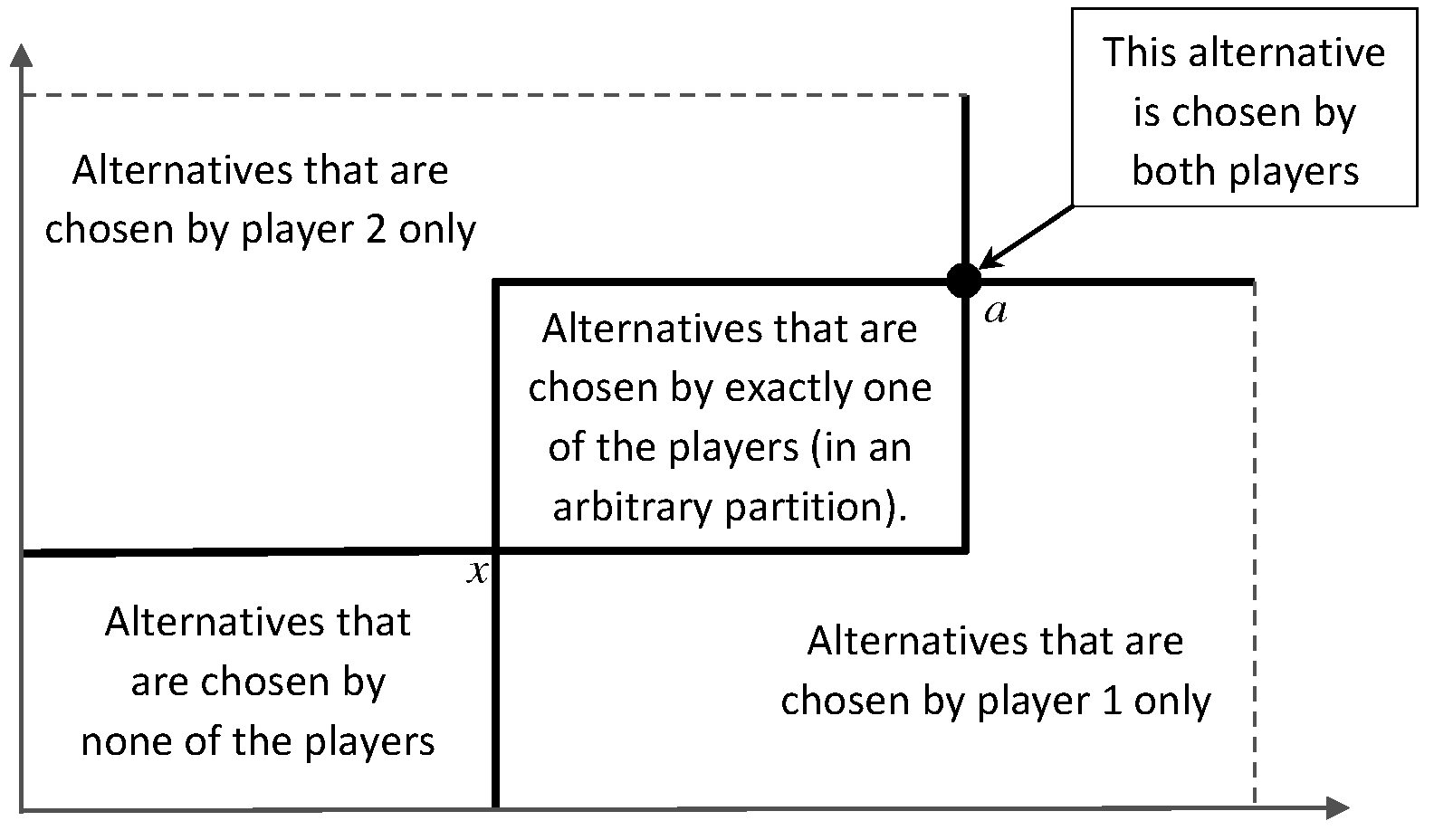}
\end{center}
\caption{An agreement equilibrium profile.}\label{pic:ag-profile}
\end{figure}

\textbf{Disagreement equilibium}, which exist if there exists a Pareto efficient average fixed point $x$. The equilibrium action profile is demonstrated in Figure \ref{pic:dis-profile}.

\begin{figure}[h]
\begin{center}
\includegraphics[scale=1]{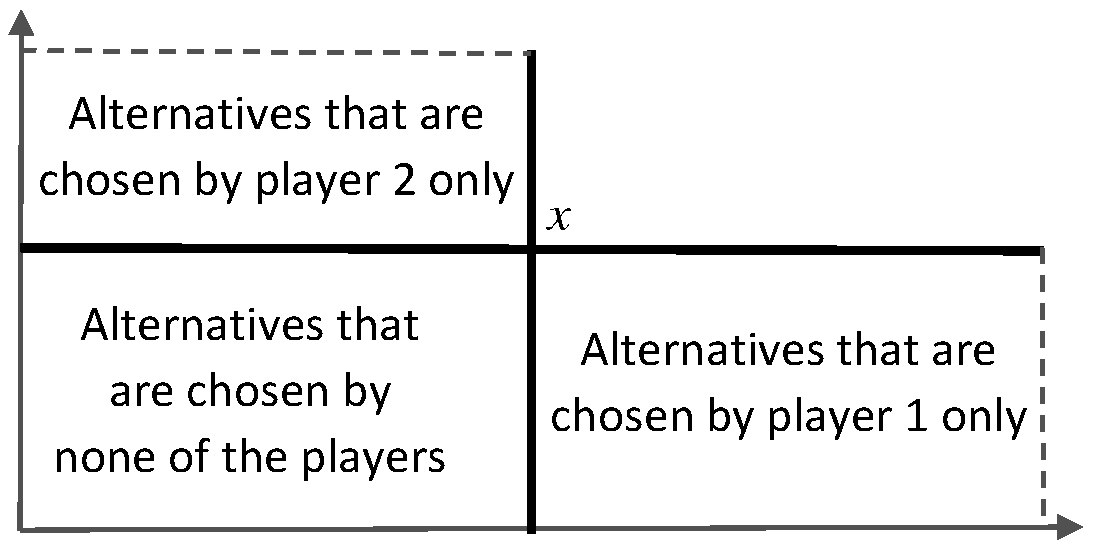}
\end{center}
\caption{A disagreement equilibrium profile.}\label{pic:dis-profile}
\end{figure}

\subsection{Bargaining}\label{sec:bar}
As we have seen, the suggested approval voting mechanism is symmetric, it implements only Pareto efficient outcomes, and in addition the set of outcomes is scale invariant (because Nash equilibria outcomes are scale invariant). Efficiency, symmetry and scale invariance are universally accepted axioms in bargaining theory, see~\cite{KS}. Therefore, it is valid to consider the suggested mechanism in the context of bargaining problems with a finite set of alternatives (this model has been previously addressed, see~\cite{A}). Implementation in the context of bargaining has been mainly studied for sub-game perfect equilibria implementation, see~\cite{BRW,Ma,H,A,Mi}, which establish strong positive results on the implementation of classical bargaining solutions. Nash implementation, on the other hand, has received less attention in the literature. One evident reason is that the above mentioned strong positive results are \emph{impossible} for Nash implementation, see~\cite{Mi}.

The suggested approval voting mechanism has multiple equilibria outcomes and therefore cannot define a unique point as a solution for the bargaining problem. In fact, there exists no mechanism with unique equilibrium outcome that satisfies approximate Pareto efficiency, as we show in Proposition \ref{pro:imp-u} in the Appendix. Nevertheless, the mechanism does produce \emph{some} prediction regarding the possible outcomes of the bargaining interaction: no player gets a payoff that is less than his utility at the minimal average fixed point. Here we demonstrate this observation in the classical ``partition of the pie" problem (see e.g., Rubinstein~\cite{R}). A similar exercise can be done for general bargaining sets.
In the pie partition problem, there is a unit of good that should be split among the bargainers. There are several ways of modelling this problem with a finite number of alternatives:
\begin{enumerate}
\item The collection is $A=((0,1),(1,0))$.

\item The collection is a $\frac{1}{k}$-grid of the line $\conv((0,1),(1,0))$. Namely, $A=((\frac{c}{k},1-\frac{c}{k}))_{c=0}^k$.
\end{enumerate}
There is also a modelling which does not assume efficiency in the definition of the problem.
\begin{enumerate}
\setcounter{enumi}{2}
\item The collection is a $\frac{1}{k}$-grid of the triangle $\conv((1,0),(0,1),(0,0))$. Namely, 
\begin{align}\label{eq:sp-pie}
A=\{(\frac{c}{k},\frac{d}{k}): c,d\in \mathbb{N} \text{ and } \frac{c}{k}+\frac{d}{k}\leq 1 \}.
\end{align} 
\end{enumerate}

For both modellings 1 and 2 the unique outcome of both mechanisms $\mathcal{AV}$ and $\mathcal{AV}_\delta$ is the point $(\frac{1}{2},\frac{1}{2})$. This is not very interesting. Actually $(\frac{1}{2},\frac{1}{2})$ is the unique outcome of \emph{every} anonymous mechanism $M$ for both modellings 1 and 2. This follows from the fact that the game induced by an anonymous mechanism, is a symmetric game with constant sum 1, and therefore has a unique Nash equilibrium outcome $(\frac{1}{2},\frac{1}{2})$.

It is interesting to analyse the outcomes of the mechanism $\mathcal{AV}_\delta$ for small values of $\delta$ and $k$ in modelling (3). The following proposition states that the outcomes get close to the Pareto efficient segment that connects the points $(0.39,0.61)$ and $(0.61,0.39)$. 

\begin{proposition}\label{pro:pie}
Let $A=A(k)$ be the collection in equation \eqref{eq:sp-pie}. For every $k$ and every $\delta>0$ all the pure Nash equilibria outcomes of the mechanism $\mathcal{AV}_\delta$ are $(\delta+\frac{1}{k})$-close to the segment $\conv((x,1-x),(1-x,x))$, where $x\approx 0.39$ is the solution of the equation $x^3-x+\frac{1}{3}=0$ in the segment $x\in [0,\frac{1}{2}]$.
\end{proposition}

Before the proof of the Proposition, we present a simple observation (given Lemma \ref{lem:afp-pd}) regarding average-fixed points that will be useful for finding the average fixed points of $A(k)$.

\begin{lemma}\label{lem:sym}
Let $A$ be a symmetric collection of utility profiles, and let $x=(x_1,x_2)$ be an average fixed point of $A$, then $x_1=x_2$. 
\end{lemma}

\begin{proof}
By symmetry of the collection, $(x_2,x_1)$ is also an average fixed point. By Lemma \ref{lem:afp-pd} it must be the case that $(x_1,x_2)\leq \leq (x_2,x_1)$ or the opposite $(x_2,x_1) \leq \leq (x_1,x_2)$. In both cases it follows that $x_1=x_2$.
\end{proof}

\begin{figure}[h]
\begin{center}
\includegraphics[scale=0.9]{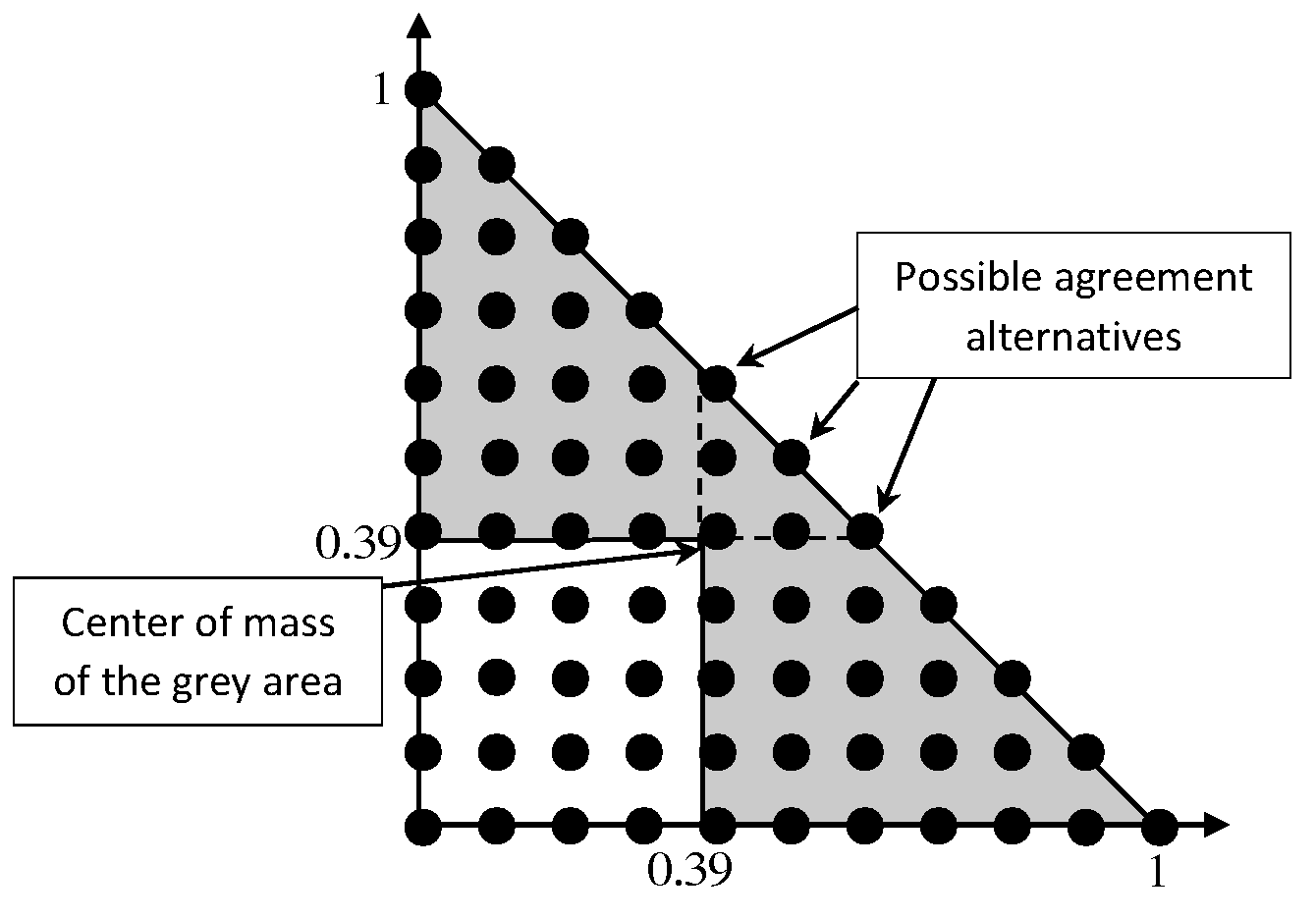}
\end{center}
\caption{Equilibria of the splitting-the-pie example.}\label{pic:pie}
\end{figure}

\begin{proof}[Proof of Proposition \ref{pro:pie}]
First we approximate the average fixed points of the set $A=A(k)$ up to an error of $\frac{1}{k}$. By Lemma \ref{lem:sym} all average fixed points of $A$ are of the form $(x,x)$ for $x\in [0,\frac{1}{2}]$. 

We consider the continuous version where we replace the sets $A\setminus A_{<,<}(x,x)$ and $A\setminus A_{\leq, \leq}(x,x)$ by the set 
\begin{align*}
\mathcal{B}= \conv((0,0),(0,1),(1,0))\setminus \conv((0,0),(0,x),(x,0),(x,x))
\end{align*}
with the uniform density. The center of mass of $\mathcal{B}$ approximates both the average of $A\setminus A_{<,<}(x,x)$ and the average of $A\setminus A_{\leq, \leq}(x,x)$ up to an error of $\frac{1}{k}$ because the difference between these expressions depends only the boundary points that are close to $A_{\leq, \leq}(x,x)\setminus A_{<,<}(x,x)$ which are at most $\frac{1}{k}$ fraction of all points in $A\setminus A_{\leq, \leq}(x,x)$.

The center of mass of $\mathcal{B}$ is given by
\begin{align}\label{eq:cm}
\frac{1}{\frac{1}{2}-x^2}\left[\frac{1}{2}(\frac{1}{3},\frac{1}{3})-x^2(\frac{x}{2},\frac{x}{2})\right]
\end{align}
where $(\frac{1}{3},\frac{1}{3})$ is the center of mass of $\conv((0,0),(0,1),(1,0))$, $(\frac{x}{2},\frac{x}{2})$ is the center of mass of $\conv((0,0),(0,x),(x,0),(x,x))$, and $\frac{1}{2}$ and $x^2$ are the corresponding areas of these sets.
By the average fixed point assumption we deduce from formula \eqref{eq:cm} that
\begin{align*}
\frac{1}{\frac{1}{2}-x^2}\left[\frac{1}{6}-\frac{x^3}{2}\right]=x \Rightarrow x^3-x+\frac{1}{3}=0
\end{align*}
This equation has a unique solution for $x\in [0,\frac{1}{2}]$. Therefore all average fixed points $(y,y)$ of $A(k)$ are located $\frac{1}{k}$ close to $(x,x)$.
Finally, by the characterization of equilibria outcome in Theorem~\ref{theo:laa}, we get that all equilibria are agreement equilibria, where the agreement point is $(a,1-a)$ for $y\leq a\leq 1-y$ and the outcome is $\delta$ close to $(a,1-a)$. See Figure~\ref{pic:pie}.
\end{proof}

It is worth explaining why allowing sub-efficient alternatives can create a wider set of efficient bargaining solutions. The sub-efficient alternatives increase the level of punishment; i.e., player $i$ can reduce the payoff of player $3-i$ below $\frac{1}{2}$ (Figure~\ref{pic:pie}).  
Therefore, new equilibria arise where player $i$ plays a ``clever" punishment strategy in which the opponent's best option is to agree on a division where he gets less than $\frac{1}{2}$. An example of such a ``clever" punishment strategy in Figure~\ref{pic:pie} for player 1 is the strategy that includes the 22 utility profiles in the bottom-right trapezoid and one additional utility profile $(0.6,0.4)$. This strategy is ``clever" in the above sense because it balances between two opposite goals of player 1: On the one hand, to punish player 2 in order to force player 2 to agree to an unfair division; and on the other hand, to exclude alternatives that are bad for himself, because with a positive probability $\delta$ these bad alternatives are taken into account (even in the case of agreement).

\subsection{Pareto frontier}\label{sec:par-f}

The mechanism is allowed to return lottery outcomes, whereas we measured the efficiency of a mechanism with respect to the \emph{pure} utility profiles. Consider, for instance, the following collection of alternatives:
\begin{align*}
A=((1,0),(0.6,0),(0,1),(0,0.6)).
\end{align*} 
The unique average fixed points of $A$ is $(0.4,0.4)$ which is Pareto efficient (with respect to $A$). Therefore, by Theorem \ref{theo:laa} $(0.4,0.4)$ is the unique (disagreement) equilibrium outcome. It is reasonable to argue that the equilibrium outcome $(0.4,0.4)$ is not Pareto optimal, because the mechanism can choose a lottery with expected utilities $(0.5,0.5)$.

A stronger (and arguably more suitable in our settings) notion of Pareto optimality is the Pareto frontier. Given a collection $A$, an outcome $x=(x_1,x_2)$ is \emph{$\varepsilon$-close to the Pareto frontier} if there is no alternative $y\in \conv(A)$ such that $y>>x+(\varepsilon,\varepsilon)$.


We argue that the mechanism $\mathcal{AV}_\delta$ can be modified to a similar mechanism whose outcomes are arbitrarily close to the Pareto \emph{frontier} in all equilibria.

A \emph{$k$-uniform} distribution over the alternatives $[n]$ is a uniform distribution over a multiset of size $k$ of alternatives in $[n]$. We denote by $k$-$UN([n])$ the set of all $k$-uniform distributions over\footnote{Note that the number of $k$-uniform different distributions is finite, and is equal to $\binom{n+k-1}{k-1}$.} $[n]$. 

In the modified mechanism $\mathcal{AV}^k_\delta$, each player submits a set of approved $k$-uniform distributions over alternatives. Namely, each player $i=1,2$ submits a list $L_i\subset k$-$UN([n])$. The mechanism $\mathcal{AV}^k_\delta$ chooses the outcome lottery exactly in the same way as $\mathcal{AV}_\delta$ does. The only difference, is that here we have a uniform distribution over $k$-uniform distributions, which induces a distribution over alternatives. 

\begin{proposition}\label{pro:sa-pf}
The mechanism $\mathcal{AV}^k_\delta$ admits a pure Nash equilibrium for every collection $A$, and $\mathcal{AV}^k_\delta$ is an anonymous mechanism which is $\left(\delta+\frac{1}{k}\right)$-close to the Pareto frontier in all equilibria.
\end{proposition}

\begin{proof}
The mechanism $\mathcal{AV}^k_\delta$ over the collection $A$ is identical to the mechanism $\mathcal{AV}_\delta$ over the collection $k$-$UN(A)$, where $k$-$UN(A)=\{\mathbb{E}_{i\sim \mu} (a_i):\mu \text{ is a } k \text{-uniform distributoin over } [n]\}$ is the set of expected outcomes under $k$-uniform distributions over $A$. By Corollary \ref{cor:laa}, this proves existence of pure Nash equilibrium.

\begin{figure}[h]
\begin{center}
\includegraphics[scale=1]{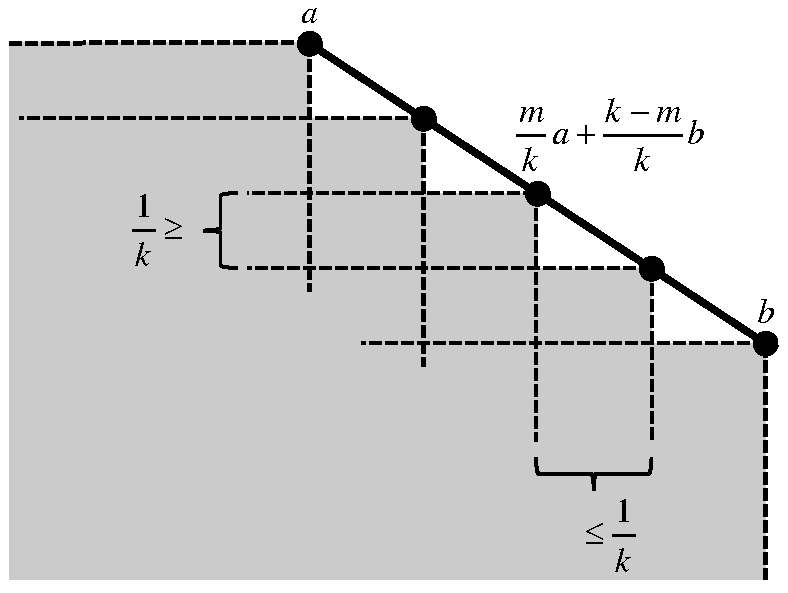}
\end{center}
\caption{Points that are Pareto dominated by $k$-uniform outcomes.}\label{pic:pf}
\end{figure}

For every line $\conv(a,b)$ on the Pareto frontier, where $a,b\in PE(A)$, the outcomes $\{ \frac{m}{k} a+\frac{k-m}{k} b \}_{m=1}^k$ are $k$-uniform distribution outcomes on the Pareto frontier. Figure ~\ref{pic:pf} demonstrates that every point that is $\frac{1}{k}$-far from 
the Pareto frontier is Pareto dominated by one of such outcomes $\frac{m}{k} a+\frac{k-m}{k} b$. Therefore $\delta$-Pareto efficiency with respect to $k$-$UN(A)$ implies $\left(\delta+\frac{1}{k}\right)$-closeness to the Pareto frontier of $A$. By Corollary \ref{cor:laa}, all equilibria of the mechanism $\mathcal{AV}_\delta$ over the collection $k$-$UN(A)$ are $\delta$-Pareto efficient (with respect to $k$-$UN(A)$), which implies that all equilibria of the mechanism $\mathcal{AV}^k_\delta$ over the collection $A$ are $\left(\delta+\frac{1}{k}\right)$-close to the Pareto frontier. 
\end{proof}

\subsection{Proof of Theorem \ref{theo:laa}}\label{sec:proof}
We start with introducing several additional notations. 
For a set of utility profiles $S\subset A$ we denote by $I_S=\{i\in [n]:a_i\in S\}$ the corresponding set of alternatives. In the opposite direction, for a set of alternatives $L$, we denote by $A_L\subset A=\{a_l\in A:l\in L\}$ the corresponding set of utility profiles. 

For a fixed-point $x$ which includes the boundary point $B\subset A_{\leq ,\leq}(x)\setminus A_{<,<}(x)$ (i.e., $\avg((A\setminus A_{\leq, \leq}(x))\cup B)=x$) we partition the boundary points in $B$ into two sets $B_i=\{a\in B:a_i=x_i\}$, where $B_1\cupdot B_2=B$.

We start with showing that every outcome $x=(x_1,x_2) \in DIS(A)$ is a disagreement equilibrium outcome. 

We split the utility profiles in $A\setminus A_{\leq, \leq}(x)$ into two groups: 
\begin{align*}
D_i=\{a \in A: a_i > x_i\} \text{ for } i=1,2.
\end{align*}
The sets $B_i$ and $D_i$ are demonstrated in Figure ~\ref{pic:dis}.

\begin{figure}[h]
\begin{center}
\includegraphics[scale=1]{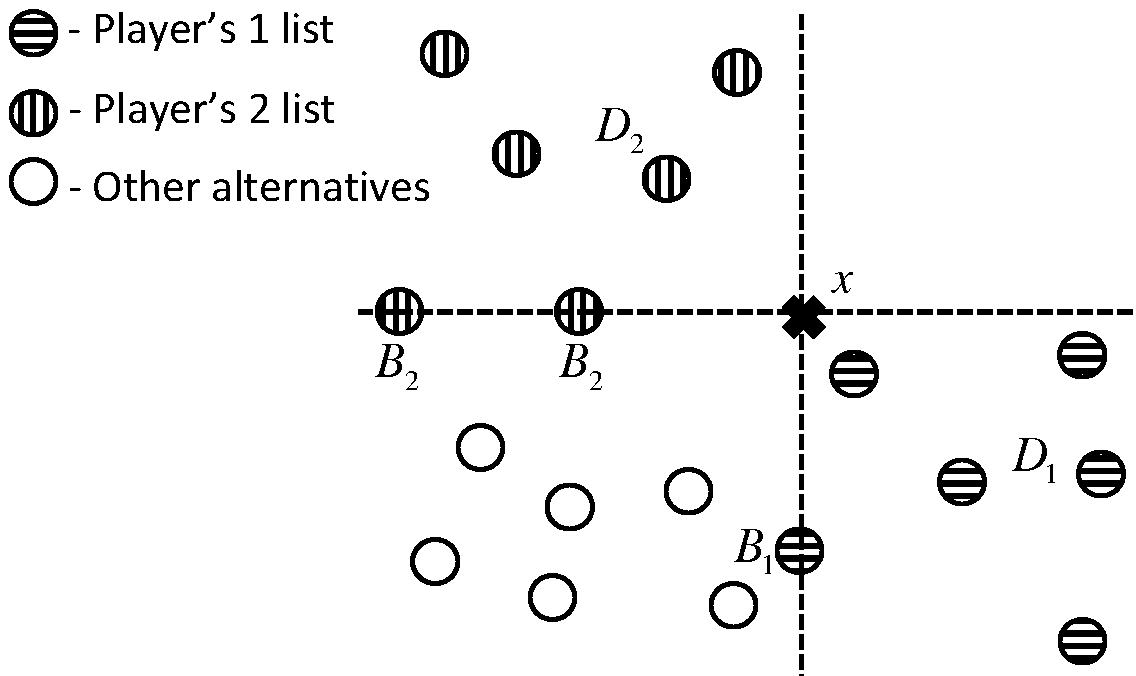}
\end{center}
\caption{A disagreement equilibrium.}\label{pic:dis}
\end{figure}

The fixed point $x$ belongs to $DIS(A)$, therefore, there is no  $a\in A$ such that $a >> x$. So the sets $B_1\cup D_1$ and $B_2 \cup D_2$ are disjoint. Therefore the payoff profile for the profile $(I_{B_1\cup D_1}, I_{B_2 \cup D_2})$ is $x$ (because $x$ is a fixed point). We argue that the action profile $(I_{B_1\cup D_1}, I_{B_2 \cup D_2})$ is a Nash equilibrium. Player $1$ approves all the above-average alternatives ($D_1$) and disapprove all the below-average alternatives. Therefore, player $1$ cannot increase his payoff by remaining in a disagreement. Note that all the alternatives in $B_2 \cup D_2$ are below-average alternatives for player 1. Therefore, every agreement will reduce the payoff of player 1. Symmetric arguments prove that player 2 has no profitable deviation.

Before we show that every outcome in $AG(A)$ is a disagreement equilibrium outcome, we introduce a Lemma that will be useful in its proof.

\begin{lemma}\label{lem:max-dis}
Let $x$ be an average fixed point and let $S_2\subset A$ be a list (of player 2) that approves all the alternatives in $A_{<,>}(x)\cup B_2$ and disapprove all the alternatives in $A_{\leq ,\leq }(x)\setminus B_2$. Then 
\begin{align*}
\max_{S_1\subset A} \avg_1(S_1 \cup S_2)=x_1.
\end{align*}
\end{lemma}
\begin{proof}
For $S_1=\{a\in A:a_1\geq x_1\}$ we have $\avg_1(S_1 \cup S_2)=x_1$, this is because $x$ is an average fixed point and every choice of the boundary points $\{a\in A:a_1=x_1\}$ does not effect $\avg_1$. This is also the \emph{maximal} value of $\avg_1(S_1 \cup S_2)=x_1$, because every disapprovement of above-average or approvement of below-average alternative will reduce the average.
\end{proof}

Now we show that every outcome $(1-\delta)a+\delta x \in AG(A)$ is an agreement equilibrium outcome. 
We denote by $R=A_{\leq, \leq}(a)\cap A_{\geq,\geq}(x)$ the utility profiles in the rectangle that is formed by the two points $a$ and $x$. We also denote $C_1=A_{>,\leq}(x_1,a_2)\setminus R$, and $C_2=A_{\leq,>}(a_1,x_2)\setminus R$.
Let $R=R_1 \cupdot R_2$ be an arbitrary partition of the utility profiles in $R$. The set $B_i,R$ and $C_i$ are demonstrated in Figure ~\ref{pic:ag}. 

\begin{figure}[h]
\begin{center}
\includegraphics[scale=0.9]{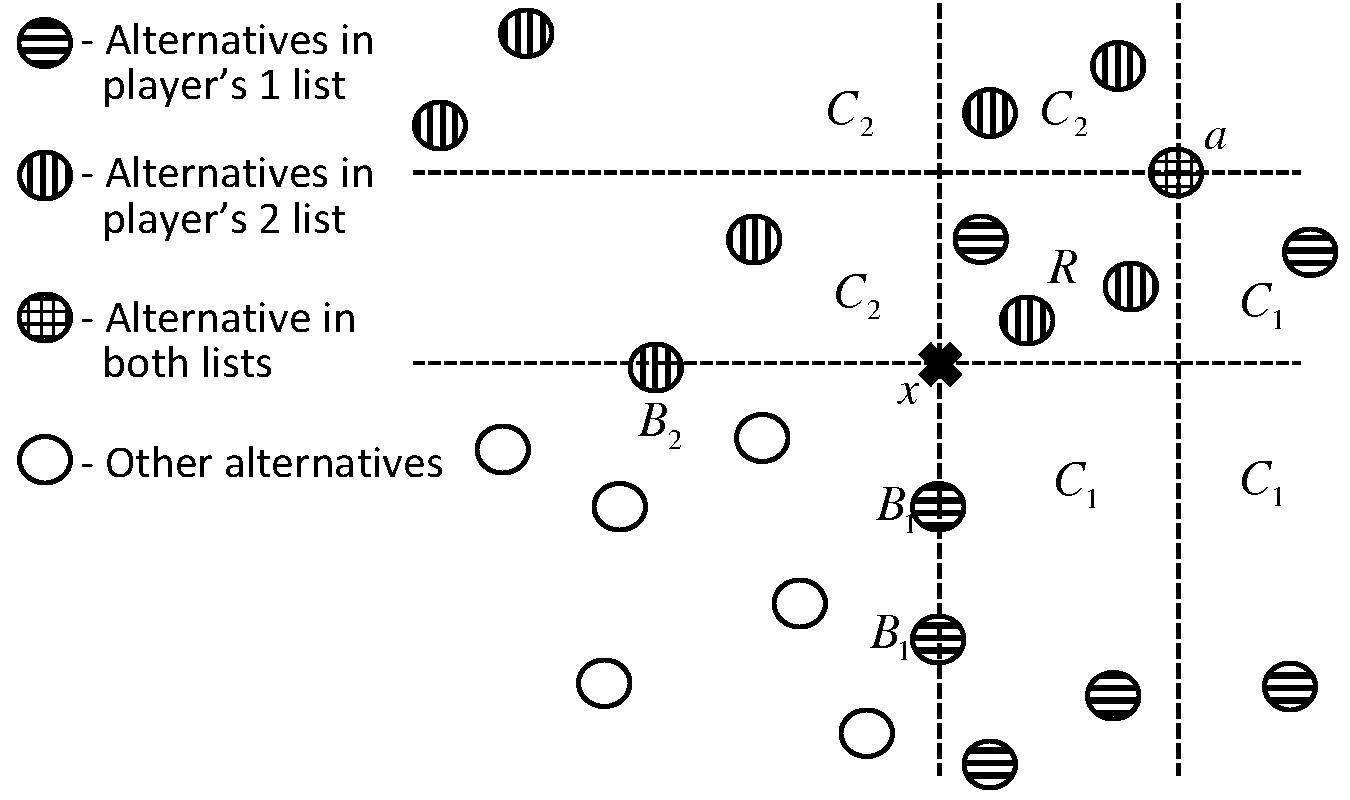}
\end{center}
\caption{An agreement equilibrium.}\label{pic:ag}
\end{figure}

We argue that the action profile $(L_1,L_2)=(I_{B_1 \cup C_1 \cup R_1 \cup \{a\}}, I_{B_2 \cup C_2 \cup R_2 \cup \{a\}})$ is an agreement equilibrium. First it is easy to check that $L_1 \cap L_2 = \{a\}$ and $L_1 \cup L_2= I_{(A\setminus A_{\leq, \leq}(x))\cup B}$, therefore the outcome is indeed $o=(1-\delta)a+\delta x$. Player 1 cannot improve the ``disagreement" payoff $x_1$ in $o_1$, because all the above-average alternatives are approved by one of the players (namely, belongs to $L_1\cup L_2$), not necessarily by player 1. Moreover, all the below-average alternatives are disapproved by player 1 (namely not in $L_1$). Now we show that the ``agreement" payoff $a_1$ cannot be improved by a unilateral deviation of player 1. Player $1$ can break the agreement (i.e., to disapprove $a$), by Lemma \ref{lem:max-dis} this will reduce his payoff to $x_1$ or less. Finally, player 1 cannot improve the agreement payoff $a_1$ by switching to another (or adding an additional) agreement alternative because for all $b\in L_2$, $b_1\leq a_1$. Symmetric arguments prove that player 2 has no profitable deviation.

Now we turn to the second part of the proof where we show that the constructed above equilibria are \emph{all} the pure Nash equilibria of the game.

We start with showing that in every agreement equilibrium the agreement outcome is unique:
\begin{lemma}\label{lem:unique}
Let $(L_1,L_2)$ be a pure Nash equilibrium such that $L_1\cap L_2 \neq \emptyset$, then for every $i,j \in L_1\cap L_2$, $a_i=a_j$.
\end{lemma}

\begin{proof}
Assume that $L_1\cap L_2\neq \emptyset$ and the intersection contains alternatives with different utility profiles. Without loss of generality, we assume that not all outcomes of player 1 are identical at the alternatives $L_1\cap L_2$, and that $i\in L_1\cap L_2$ obtains the minimal utility for player 1. Player 1 will gain by a disapproval of alternative $i$ (i.e., $L'_1=L_1\setminus \{ i\}$) because $\avg_1(A_{L'_1\cup L_2})=\avg_1(A_{L_1\cup L_2})$ (because the union remains unchanged) and $\avg_1(A_{L'_1\cap L_2})>\avg_1(A_{L_1\cap L_2})$. 
\end{proof}
Now we claim that for every equilibrium the union of the approved sets forms an average fixed point.
\begin{lemma}\label{lem:union}
Let $(L_1,L_2)$ be a pure Nash equilibrium and let $x=\avg(A_{L_1 \cup L_2})$, then $x$ is an average fixed point; i.e., $A\setminus A_{\leq, \leq}(x)\subset A_{L_1 \cup L_2} \subset A \setminus A_{<,<}(x)$.
\end{lemma}

\begin{proof}
If there exists an alternative $l$ such that $l\notin L_1\cup L_2$ and $a^l_i>x_i$ then player $i$ has a profitable deviation to $L_i\cup \{l\}$. Therefore, $A\setminus A_{\leq ,\leq}(x)\subset A_{L_1 \cup L_2}$.

If there exists an alternative $l\in L_1\cup L_2$ such that $a_l<<x$ then we consider two cases.

\begin{itemize}
\item[Case 1:] $l \in L_i$ but $l \notin L_{3-i}$. Then player $i$ has a profitable deviation to $L_i\setminus \{l\}$.
\item[Case 2:] $l\in L_1\cap L_2$. Then by Lemma~\ref{lem:unique} $a_l$ is the unique agreement outcome. Player 1 has a profitable deviation to a disagreement by excluding the alternative $l$ (and all the identical alternatives $k$ such that $a_k=a_l$) from his approval set.
\end{itemize} 
Therefore $A_{L_1 \cup L_2} \subset A \setminus A_{<,<}(x)$.
\end{proof}
The following lemma shows that the agreement outcome (if exists) is better than the disagreement outcome, and it is efficient:
\begin{lemma}\label{lem:ag}
Let $(L_1,L_2)$ be a pure Nash equilibrium such that $L_1\cap L_2 \neq \emptyset$ and let $a\in A_{L_1\cap L_2}$, then
\begin{enumerate}
\item $a\geq \geq \avg(A_{L_1 \cup L_2})$.
\item $a\in PE(A)$.
\end{enumerate} 
\end{lemma}
\begin{proof}
1. Assume to the contrary that $a_1<\avg_1(A_{L_1 \cup L_2})$. By Lemma \ref{lem:unique} $a$ is the unique agreement outcome. Player 1 has a profitable deviation to a disagreement by excluding the alternative $I_a$ (and all the identical alternatives $I_b$ that are played by player 2) from his set. A similar argument excludes the possibility of $a_2<\avg_2(A_{L_1 \cup L_2})$.

2. Assume to the contrary that there exists $a'>>a$. By (1) we know that $a'>> a\geq \geq \avg(A_{L_1 \cup L_2})$. By Lemma \ref{lem:union} we get that $a'\in A_{L_1 \cup L_2}$. Without loss of generality we assume $a'\in L_2$. Then player $1$ can increase his payoff by including $a'$ into the set of agreements.
\end{proof}

The following lemma shows that a disagreement equilibrium has to be efficient:
\begin{lemma}\label{lem:dis}
Let $(L_1,L_2)$ be a pure Nash equilibrium such that $L_1\cap L_2 = \emptyset$, then there is no $a\in A$ such that $a>>\avg(A_{L_1 \cup L_2})$.
\end{lemma}
\begin{proof}
Assume to the contrary that there exists $a>> \avg(A_{L_1\cup L_2})$. By Lemma \ref{lem:union} we get that $a\in A_{L_1 \cup L_2}$. Without loss of generality we assume $a\in L_2$. Then player $1$ can increase his payoff by adding $a$ into his set and turning the disagreement into agreement on $a$.
\end{proof}
Summarizing, every disagreement equilibrium outcome is an average fixed point (Lemma \ref{lem:union}) and is not dominated by any utility profile (Lemma \ref{lem:dis}) which restricts the set of outcomes to be in $DIS(A)$.  In every agreement equilibrium the agreement is on a unique outcome (Lemma \ref{lem:unique}) which is Pareto efficient (Lemma \ref{lem:ag}). In addition, the set $L_1\cup L_2$ forms an average fixed point (Lemma \ref{lem:union}), and $\avg(L_1\cup L_2)$ is Pareto dominated by the agreement outcome (Lemma \ref{lem:ag}). This restricts the set of outcomes to be in $AG(A)$.

\section{Weighted approval voting mechanisms}\label{sec:non-unif}

The approval voting mechanism chooses an alternative \emph{uniformly} at random from the intersection, or from the union. For the proofs of the results, it is sufficient to assume that the the distribution by which the mechanism chooses the alternative (from the intersection or from the union) assign \emph{positive probability} to each alternative.

For instance we can consider a mechanism where each alternative $i$ has a weight $w_i>0$. In case the sets are disjoint ($L_1\cap L_2 =\emptyset$), the mechanism chooses the alternative $j\in L_1\cup L_2$ with probability $w_j/\sum_{i\in L_1\cup L_2} w_i$. In case the sets intersect ($L_1\cap L_2 \neq \emptyset$), the mechanism chooses the alternative $j\in L_1\cup L_2$ with probability $\delta w_j/\sum_{i\in L_1\cup L_2} w_i$, and it chooses an alternative $j\in L_1\cap L_2$ with probability $(1-\delta) w_j/\sum_{i\in L_1\cap L_2} w_i$. We denote this mechanism by $\mathcal{AV}_\delta^w$ where $w=(w_1,...,w_n)$.

We argue that all the arguments that hold for $\mathcal{AV}_\delta$ can be translated to arguments on $\mathcal{AV}_\delta^w$. This is because all the arguments rely on the fact that inclusion (exclusion) of below-average (above-average) numbers into (from) the calculation of the average, reduces the average. This basic fact remains true also in the case of weighted averages.

More formally, the analogue of an average is the weighted average 
\begin{align*}
\avg^w(B)=\frac{1}{\sum_{a^i \in B} w_i} \sum_{a^i \in B} w_i a^i
\end{align*}
The analogue of average fixed point is the notion of \emph{weighted average fixed point}.

\begin{definition}
Given collection $A$ and weights $w$, a vector $x=(x_1,x_2)$ is a \emph{weighted average fixed point} of the collection $A$ if there exists a subset $B\subset A_{\leq, \leq}(x)\setminus A_{<,<}(x)$ such that $\avg^w((A\setminus A_{\leq, \leq}(x)) \cup B)=x$.
We denote by $AFP^w(A)$ the set of weighted average fixed points. 
\end{definition}
Similarly, we can define the notion of minimal weighted average fixed point.
\begin{definition}
A point $m^w_{afp}(A) \in [0,1]^2$ is a minimal weighted average fixed point if $m^w_{afp} \in AFP^w(A)$, and for every $y\in AFP^w(A)$ holds $y \geq \geq m^w_{afp}$.
\end{definition}
Using similar arguments to those in Section \ref{sec:afp}, we can prove that minimal average fixed point exists and is unique. 

We have an analogous characterization of equilibrium outcomes for the case of weighted approval voting mechanisms.

\begin{theorem}
For every collection $A$, every $0< \delta \leq 1$ and every weight vector $w$, the set of pure Nash equilibria outcomes of the mechanism  $\mathcal{AV}^w_\delta$ is exactly the union of the following two types of equilibrium outcomes:
\begin{enumerate}
\item The set of agreement equilibrium outcomes 
\begin{align*}
AG(A)=\{(1-\delta)a+\delta x: a\in PE(A), x \in AFP^w(A), \text{ and } a\geq \geq x\}.
\end{align*}
\item The set of disagreement equilibrium outcomes 
\begin{align*}
DIS(A)=\{x\in AFP^w(A): \text{There is no } a\in A \text{ such that } a >> x\}.
\end{align*}
\end{enumerate}
\end{theorem}
The proof uses exactly the same arguments as the proof of Theorem \ref{theo:laa}, where the notion of average is replaced by the notion of weighted average.

Similarly, we can define the \emph{weighted minimal-AFP Pareto efficient correspondence} to be 
\begin{align*}
\sigma^w(A)=\{a\in PE(A \cup AFP^w(A)): a\geq \geq m^w_{afp}(A)\}.
\end{align*}
The analogue of Corollary \ref{cor:scc} for weighted approval voting mechanism states the following.
\begin{corollary}\label{cor:wscc}
The mechanism $\mathcal{AV}^w_\delta$ virtually implements the social choice correspondence $\sigma^w$.
\end{corollary}

This corollary has a similar analogue in the implementation literature. We normalize $w$ to be a probability vector. We define the social choice correspondence
\begin{align*}
\tau^w(A)=\{x\in PE(\conv(A)): x\geq \geq \avg^w(A)\}. 
\end{align*}
Note that $\sigma^w$ and $\tau^w$ are similar. In both cases the correspondence picks all the Pareto efficient points that Pareto dominate some \emph{threshold point}. For $\sigma^w$ the threshold point is the minimal weighted average fixed point, for $\tau^w$ the threshold point is simply the weighted average.  Dutta and Sen~\cite{DS} have proven that (under mild domain restrictions) $\tau^w$ can be implemented using their canonical mechanism. We prove that a discrete analogue of a similar SCC which picks as a threshold point the minimal average fixed point rather than the weighted average, can be implemented using weighted approval voting mechanism.  We also emphasize that weighted average fixed point always Pareto dominates the weighted average and therefore our mechanism implements a \emph{subset} of the outcomes that are implemented by the mechanism of Dutta and Sen~\cite{DS}.

\section{Discussion}\label{sec:dis}

\subsection{Other solution concepts}
\subsubsection{Strict Nash equilibria}\label{sec:str}
It remains an interesting open question to analyse the set of \emph{mixed} Nash equilibria of the presented approval voting mechanism. Given the fact that this question remains open it is natural to ask whether the pure Nash equilibria of the mechanism are ``more natural" solutions than mixed ones. Under mild genericity assumptions, we argue that indeed the pure Nash equilibria are more natural, because they are \emph{strict}; i.e., each players loses by deviation from equilibrium. More concretely, the genericity assumptions are:
\begin{enumerate}
\item Players' utility preferences are strict; i.e., $a_i\neq a'_i$ for $a\neq a'$.
\item We call an average fixed point $x$ \emph{trivial} if $|A \setminus A_{<,<}(x)|=1$, namely the average is done over a single alternative (this alternative must Pareto dominate all the other alternatives). The second genericity assumption requires that player is never indifferent between a non-trivial average fixed point and one of his alternatives.
\end{enumerate}
It is easy to see that under these assumptions, the pure Nash equilibria of the mechanism $\mathcal{AV}_\delta$ (see Section \ref{sec:eq-char}) are strict.

\subsubsection{Sequential elimination of dominated strategies}\label{sec:seq-elim}
A natural direction that arises is to analyse the approval voting mechanism w.r.t.\  sequential elimination of dominated strategies (rather than Nash equilibria). Unfortunately, our characterization of equilibrium strategies (see Section \ref{sec:eq-char}) indicates that there are cases where many strategies survive the procedure of sequential elimination (and therefore the sequential elimination has low predictive power). This simply follows from the fact that every strict equilibrium strategy (of a single player) survives any sequential elimination. To construct examples where the set of strict equilibria strategies (of a single player) is large, one might construct a case with ``many" alternatives that Pareto dominate an average fixed point $x$ and are Pareto dominated by an efficient alternative $a$ (i.e., the set $B:=\{b\in A: x<<b<<a\}$ is large), see Figure \ref{pic:ag-profile}. In such a case \emph{every subset} of alternatives from $B$ can be ``completed" to a strict equilibrium strategy of player 1 by adding the corresponding list of strategies out of $B$ (which is the set $\{b\in A: b_1>x_1 \text{ and } b_2<a_2 \}\setminus B$). Therefore every such ``completed" strategy survives the elimination of dominated strategies. 
Note that these completed strategies are not necessarily \emph{sincere} strategies, which are strategies of the form  $\{b\in A:b_1>c\}$ for some constant $c$. In the following example in all equilibria at least one player does not play a sincere strategy in equilibrium. This example also demonstrates that existence of a sincere pure Nash equilibrium---which is guaranteed for the standard approval voting mechanism (see \cite{NL})---is not guaranteed in our modified approval voting. 

\begin{example}\label{ex:sin}
There are six alternatives with utilities
\begin{align*}
A=((9,0),(0,9),(8,8),(7,7),(6,0),(0,6)).
\end{align*}
The unique average fixed point of $A$ is $(5,5)$. According to our characterization of pure Nash equilibria, the game admits exactly two pure Nash equilibria:
\begin{align*}
(\{1,3,4,5\},\{2,3,6\}) \text{ and } (\{1,3,5\},\{2,3,4,6\}).
\end{align*}
Note that $\{2,3,6\}$ is not sincere strategy of player $2$ (because it contain the alternative $(0,6)$ but not $(7,7)$), and $\{1,3,5\}$ is not sincere strategy of player $1$.
\end{example}
 
\subsection{More than two players}\label{sec:3p}
It is natural to ask whether a similar modification of the approval voting mechanism succeeds in selecting only Pareto efficient outcomes for the case of more than two players. In the case of more than two players it is not clear how to define the weights on the alternatives which got less approvals than the maximum. Here we suggest a natural extension of the two-player mechanism, but provide a three-player counter-example for this extension. The same counter-example holds for other extensions as well.

Given an action profile, let $k$ be the maximal number of approvals (among all the alternatives). With probability $\frac{\delta^m}{\sum_{i=1}^k \delta^k}$ the mechanism chooses uniformly at random among the alternatives that have been approved by at least $k-m$ players (note that indeed the alternatives that have been approved by $k$ players are chosen with the highest probability of $\frac{1}{\sum_{i=1}^k \delta^k}$).

The counter-example consists of eight alternatives 
\begin{align*}
A=((24,0,0),(0,24,0),(0,0,24),(7,7,0),(7,0,7),(0,7,7),(5,5,5),(6,6,6)).
\end{align*}
Consider the action profile $(\{1,4,5,7\},\{2,4,6,7\},\{3,5,6,7\})$ where every player approves all the alternatives with non-zero payoff except of the alternative $(6,6,6)$. Note that the payoffs of the players at this action profile are $(5\pm O(\delta), 5\pm O(\delta),5\pm O(\delta))$ which is Pareto dominated by the alternative $(6,6,6)$. We introduce here the intuition for the claim that the presented action profile is a pure Nash equilibrium. The average payoff of player 1 over the alternatives that achieve at least 1/2/3 approvals is 6.14/4.75/5 correspondingly. By approving the alternative $(6,6,6)$ player 1 will reduce his average payoff at alternatives that achieve at least 1 approval, while the other averages (of at least 2 and 3 approvals) remain unchanged. By disapproving the alternative $(5,5,5)$ player 1 reduces his payoff to $4.75\pm O(\delta)$ because after such deviation the maximal number of approvals is 2. It  can be checked that all other deviations (including more complex deviations that simultaneously approve and disapprove several alternatives) reduce player's 1 payoff. By symmetry the same holds for players 2 and 3.

\subsection{Ordinal preferences}\label{sec:ord}
Throughout the paper we have assumed that players have cardinal von Neumann - Morgenstern preferences over the alternatives. In fact, the proof of Theorem \ref{theo:laa} uses only the following properties of player's $i$ preferences $\succ_i$ over the lotteries:
\begin{enumerate}
\item $UN(B\cup \{a\})\succ_i UN(B)$ for $a\succ_i UN(B)$; namely, the player prefers to add an above-average alternative to the list.
\item $UN(B\setminus \{a\})\succ_i UN(B)$ for $UN(B) \succ_i a$; namely, the player prefers to erase a below-average alternative from the list.
\end{enumerate}
These two properties hold for more general settings than von Neumann - Morgenstern preferences. We can assume that preferences over lotteries are \emph{ordinal} and \emph{monotonic}. Monotonicity assumes that shifts in probability mass from less preferred to strictly preferred lotteries over $A$ yield a lottery which is strictly preferred. Note that monotonicity implies properties (1) and (2). Therefore, the results in the paper can be generalized to the case where preferences are ordinal and monotonic.

\appendix
\section{Virtual Implementation without Domain Restrictions}
Abreu and Sen \cite{AS} prove their characterization for virtual implementation under mild domain restrictions. In particular, they assume that a player's preferences over the alternatives are \emph{strict}; i.e., for $i=1,2$ holds $a_i\neq a'_i$ for every $a\neq a'$. Without the strictness restriction, as is the setting in our paper, we have the following negative result.

\begin{proposition}\label{pro:imp-u}
There is no Pareto efficient social choice \emph{function} that is virtually Nash implementable, even for two alternatives.
\end{proposition}

\begin{proof}
We set $\varepsilon<\frac{1}{2}$, and we prove that every mechanism that satisfies existence of $\varepsilon$-Pareto efficient equilibrium must violate uniqueness of equilibrium outcome.

Informally, the idea is to consider the utility profiles $((0,1),(1,1))$ where player 2 is indifferent between the two alternatives. Player 2 may act as if the utility profiles are $((0,0),(1,1))$, which in an approximately efficient mechanism should result in a high weight for the second action. Or, Player 2 may act as if the utility profiles are $((0,1),(1,0))$, which should result in a lower weight to the second action. This leads to two different equilibria outcomes.

Formally, for the utility profiles $A=((0,0),(1,1))$, let $(z_1,z_2)$ be a Nash equilibrium with the utility outcome $(p,p)$. Note that $p>\frac{1}{2}$ by $\varepsilon$-Pareto efficiency. In terms of lotteries, it means that alternative 2 is chosen with probability $p>\frac{1}{2}$ and no player can increase this probability by a unilateral deviation.

For the utility profiles $A'=((0,1),(1,0))$, let $(z'_1,z'_2)$ be a Nash equilibrium with the utility outcome $(q,1-q)$ be the outcome at the equilibrium $(z'_1,z'_2)$. Without loss of generality we assume that $q\leq \frac{1}{2}$. In terms of lotteries, it means that alternative 2 is chosen with probability $q\leq \frac{1}{2}$ and player 1 cannot increase this probability by deviation.

Now consider the utility profiles $A''=((0,1),(1,1))$. For the action profile $(z_1,z_2)$, the utility outcome is $(p,1)$, and player 1 cannot increase the probability of the second alternative to be chosen; i.e., player 1 has no profitable deviation. Obviously, player 2 has no profitable deviation either. For the action profile $(z'_1,z'_2)$, the outcome is $(q,1)$, and player 1 cannot increase the probability of the second alternative to be chosen. Therefore we have two different equilibria outcomes $(p,1)$ and $(q,1)$.
\end{proof}

\end{document}